\RequirePackage{fix-cm}
\documentclass[smallextended]{svjour3}       
\smartqed  
\usepackage{graphicx}

\usepackage{amsmath,amssymb,graphicx,mathrsfs,url}
\usepackage{tikz}

\begin{document}

\title{Operator based approach to $\mathcal{PT}$-symmetric problems on a wedge-shaped contour
}

\titlerunning{$\mathcal{PT}$-symmetric problems on a wedge-shaped contour}

\author{ Florian Leben      \and
         Carsten Trunk
}

\authorrunning{F.\ Leben and C.\ Trunk}

\institute{F. Leben \at
              Institut f\"ur  Mathematik,  Technische Universit\"{a}t Ilmenau \\
              Postfach 100565, D-98684 Ilmenau,  Germany\\
              \email{florian.leben@tu-ilmenau.de}
                \and
                 C. Trunk \at
              Institut f\"ur  Mathematik,  Technische Universit\"{a}t Ilmenau \\
              Postfach 100565, D-98684 Ilmenau,  Germany\\
              Tel.: +49 3677 69-3253\\
              \email{carsten.trunk@tu-ilmenau.de}
}

\date{Received: date / Accepted: date}

\maketitle

\begin{abstract}
We consider a second-order differential equation
$$
-y''(z)-(iz)^{N+2}y(z)=\lambda y(z), \quad z\in \Gamma
$$
with an eigenvalue parameter $\lambda \in \mathbb{C}$. In $\mathcal{PT}$ quantum mechanics $z$ runs through a complex contour $\Gamma\subset \mathbb{C}$, which is in general not the real line
nor a real half-line. Via a parametrization we map the problem back to the real line and obtain two differential equations on $[0,\infty)$ and on $(-\infty,0].$ They are coupled in zero by boundary conditions and their potentials are not real-valued.

The main result is a classification of this problem along the
well-known limit-point/ limit-circle scheme for complex potentials
introduced by A.R.\ Sims 60 years ago. Moreover, we associate
operators to the two half-line problems and to the full axis
problem and study their spectra.

\end{abstract}

\section{Introduction}
\label{intro}
In classical quantum mechanics Hamiltonians are Hermitian. Recently this has been questioned to be too restrictive. In 1998 C.M.~Bender and S.~Boettcher in the pioneering work \cite{BB98} noticed that a large class of non-Hermitian Hamiltonians possesses real spectra and suggested to construct a non-Hermitian quantum mechanic, see \cite{BB98,BBJ02,BK08,S02} or for an overview \cite{B07,B16,M10}. They adopted all axioms of quantum mechanics except the one that restricted the Hamiltonian to be Hermitian. Instead, one assumes the Hamiltonian to satisfy $\mathcal{PT}$-symmetry.
In \cite{BB98} they consider a non-Hermitian Hamiltonian corresponding to
\begin{align}\label{HAM1}
p^2-(iz)^{N+2},\quad z \in \Gamma
\end{align}
where $N$ is a natural number greater than zero. Contrary to classical quantum mechanics, $z$ runs along a complex contour $\Gamma$. For $N=0$ this Hamiltonian can be considered as a complex deformation of the classical harmonic oscillator.

Hamiltonians of the form \eqref{HAM1} are not Hermitian, but possess an antilinear $\mathcal{PT}$-symmetry, which is the combined invariance under simultaneous spatial reflection $\mathcal{P}$ and time reversal $\mathcal{T}$. The condition that the Hamiltonian is $\mathcal{PT}$-symmetric is a physical condition, because $\mathcal{P}$ and $\mathcal{T}$ both are elements of the homogenous Lorentz group of Lorentz boost and spatial rotation. Nowadays there are a lot of papers in diverse research areas about $\mathcal{PT}$-symmetric Hamiltonians, see \cite{B15,B16,B08,BK08,GMKMRC18,MGCM11,M16,S02,RMGCSK10}. E.g., a close relation to metamaterials was discovered as $\mathcal{PT}$-symmetric operators are capable to incorporate negative permittivity and permeability, cf.\ \cite{GMKMRC18,MGCM11,M16}.

In general one can not expect that the Hamiltonian \eqref{HAM1} is Hermitian in the Hilbert space $L^2$ and has real spectrum. However, in e.g.\ \cite{B07,BB98,BBJ02,DDT}, Hamiltonians with complex potential and real spectra were discussed.

In \eqref{HAM1} the contour $\Gamma$ is located in regions
 of the complex plane, such that the eigenfunctions $\phi:\Gamma\rightarrow \mathbb{C}$ of \eqref{HAM1} vanish exponentially as $|z|\rightarrow \infty$ along $\Gamma$. The regions in the complex plane where the solutions of \eqref{HAM1} vanish exponentially are wedges, which are called \emph{Stokes wedges}. Stokes wedges correspond to sectors in the complex plane. The opening angle and, hence the number of wedges, correspond only to the number $N$,
 for details we refer to Figure \ref{Figg2} below. They are bounded by lines, the
so called \emph{Stokes lines}, cf.~\cite{B07,BB98,BBJ02}. Both, Stokes wedges and Stokes lines are symmetric to the action of $\mathcal{PT}$.

It is our main aim to relate this Stokes wedge/Stokes line
dichotomy to the classical limit point/limit circle classification
from the Sturm-Liouville theory with complex potentials.

For simplicity, we choose here the special contour (cf.~\cite{AT14})
\begin{align*}
\Gamma:=\left\{z=xe^{i \phi sgn(x)}: x \in \mathbb{R}\right\},\quad \phi \in (-\pi/2,\pi/2),
\end{align*}
see Figure \ref{fig:fig_1}, and
\begin{figure}[h]
\begin{center}
\begin{tikzpicture}[domain=-1:4]
\draw[->] (-4.2,0) -- (4.2,0) node[right] {$\mathrm{Re}$}; \draw[->] (0,-0.7)
-- (0,4.2) node[above] {$\mathrm{Im}$};
\foreach \x in {-4,-3,-2,-1,1,2,3,4} \draw (\x,-.1) -- (\x,.1)
node[below=4pt] {$\scriptstyle \x$}; \foreach \y in
{1,2,3,4} \draw (-.1,\y) -- (.1,\y) node[left=4pt]
{$\scriptstyle \y$}; \draw[-,thick] (0,0) -- ( 4, 4)
node[right] {$\Gamma$};  \draw[-,thick] (0, 0) -- (-4 ,4);
\draw [-,thin, opacity =1,domain=0:45] (0,0) -- plot ({2*cos(\x)}, {2*sin(\x)}) -- (0,0) -- cycle node[right=40, above=3pt] {$\phi$};
\end{tikzpicture}\\
\end{center}
\caption{Contour $\Gamma$ in the complex plane with opnening angle $\phi$}
\label{fig:fig_1}
\end{figure}
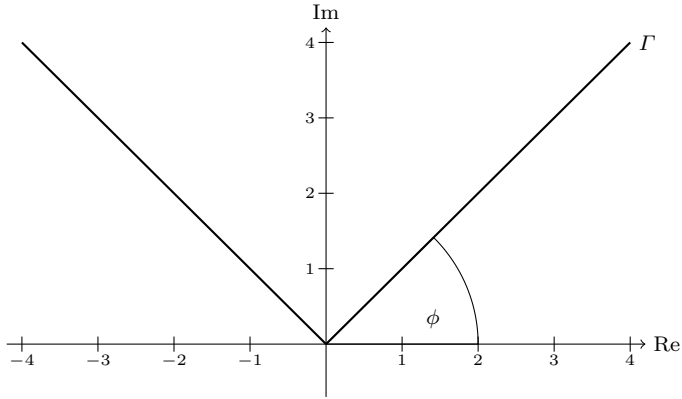
treat this problem via a Sturm-Liouville approach. Namely \eqref{HAM1} leads to the associated eigenvalue equation
\begin{align}\label{HAMEQ}
-y''(z)-(iz)^{N+2}y(z)=\lambda y(z), \quad z \in \Gamma.
\end{align}
Via the parametrization $z(x):=xe^{i \phi sgn(x)},$ $x\in \mathbb{R},$ we obtain Sturm-Liouville differential equations on $[0,\infty)$ and on $(-\infty,0]$, respectively,
\begin{align}\label{SLEE}
\tau_+w(x):=-e^{- 2i \phi}w''(x)-(ix)^{N+2}e^{(N+2)i\phi}w(x)&=\lambda w(x), \, x \in \mathbb{R}_{+}\\ \label{SLEE1}
\tau_-w(x):=-e^{ 2i \phi}w''(x)-(ix)^{N+2}e^{-(N+2)i\phi}w(x)&=\lambda w(x), \, x \in \mathbb{R}_{-}.
\end{align}

It is our aim to treat \eqref{SLEE} and \eqref{SLEE1} from an operator based perspective. This is new compared with the above cited literature from theoretical physics.

Equations \eqref{SLEE} and \eqref{SLEE1} correspond to a Sturm-Liouville problem $-(py')'+qy=\lambda y$ with
\textbf{non-real} $p$ and \textbf{non-real} $q$ on a half-axis.
But, before we consider this case, we recall
the classical Sturm-Liouville theory on a half axis (see \cite{H,Weyl10}) for \textbf{real-valued} coefficients $p$, $q$ and regular end-point $0$. Classical Sturm-Liouville theory
for $p,q$ \textbf{real} follows the following (rough) scheme:
\begin{itemize}
\item[(a)] Determine the number of $L^2$-solutions of $-(py')'+qy=\lambda y$ for $\lambda \in \mathbb{C}  \backslash \mathbb{R}$. According to the famous Weyl alternative we obtain either
one or two linearly independent $L^2$-solutions. The corresponding situation
is then called
 the limit-point case (in case of one solution)
or the limit-circle (two solutions).
\item[(b)] Define minimal and maximal operator corresponding to the differential expression $-(py')'+qy.$ Roughly speaking, the elements in the domain of the minimal vanishes at the endpoint
zero and the elements in the domain of the maximal operator satisfy no boundary conditions.
\item[(c)] Show that the minimal operator is symmetric and its adjoint is the maximal operator.
\item[(d)] Describe all self-adjoint extensions $A_{\theta}$ of the minimal operator  via a suitable parameter $\theta$ and solve the spectral problem $A_{\theta}y=\lambda y.$
\end{itemize}
This scheme is successfully used since the seminal paper of A.~Weyl \cite{Weyl10} and lead to the still very active mathematical research area of extension theory, see, e.g., the
monographs \cite{DS,EE,K,RS,Z}.

An analogous theory was subsequently developed for \textbf{non-real} potentials $q$ by A.R.\ Sims \cite{SIMS}. In a first step,
item (a) was generalized by A.R.\ Sims \cite{SIMS} to $\mathrm{Im}
\, q \leq 0$. It states that there exists at least one solution of \eqref{SLEE} in the weigthed space $L^2(0,\infty,\mathrm{Im}(\lambda-q))$, where  $\mathrm{Im}(\lambda-q)$ is the weight, and this solution is also in $L^2(0,\infty)$ for $\lambda$ in the upper complex plane. Contrary to the above Weyl
alternative in item (a) from above, now there are three cases possible:
\begin{itemize}
\item[1.] Limit-point I: There is (up to a constant) exactly one solution of $-(py')'+qy=\lambda y$ which is
simultaneously in $L^2(0,\infty,\mathrm{Im}(\lambda-q))$ and in $L^2(0,\infty)$.
\item[2.] Limit-point II: There is one solution in $L^2(0,\infty,\mathrm{Im}(\lambda-q))$, but all are in $L^2(0,\infty)$.
\item[3.] Limit-circle: All solutions are simultaneously
in $L^2(0,\infty,\mathrm{Im}(\lambda-q))$ and in $L^2(0,\infty)$.
\end{itemize}
The above approach from A.R.\ Sims \cite{SIMS} is restricted to
potentials $q$ with $\mathrm{Im}\, q \leq 0$. Instead, here
we use a generalisation which allows more general
 potential $q$ and a complex-valued function $p,$ cf.\ \cite{BCEP}. Again one obtains three cases, which corresponds to the above
limit-point I, II and limit-circle cases (and which are called
cases I, II and III in \cite[Theorem 2.1]{BCEP}). We use this
result to give a complete classification into limit-point/limit-circle of the two differential equations \eqref{SLEE} and \eqref{SLEE1}. This is done with the help of asymptotic analysis, cf.~\cite{E}. Depending on the location of the contour $\Gamma$
in terms of its angle, we specify limit-point I, II or limit-circle case.
 In the limit-point I case we do not need boundary conditions at $\pm \infty$, i.e.\ the functions $\phi$ of the domain fulfill $|\phi(x)| \rightarrow 0$ if $|x|\rightarrow \infty$ and if $\phi$ is a solution of \eqref{SLEE} or \eqref{SLEE1} even exponentially. So we reduce the (physical) notion of Stokes wedges and Stokes lines to the limit-point/limit-circle classification in the following way.
 \begin{center}
\begin{tabular}{lcl}
Equations \eqref{SLEE}, \eqref{SLEE1} in limit-point case I  & $\Leftrightarrow$ & $\Gamma$ lies in two Sokes wedges.\\[0.5 ex]
 Limit-point case II is never possible. && \\[0.5 ex]
Equations \eqref{SLEE}, \eqref{SLEE1} in limit-circle case  & $\Leftrightarrow$ & $\Gamma$ lies on two Sokes lines.\\
 \end{tabular}
 \end{center}
 This correspondence between $\mathcal{PT}$ quantum mechanics and
 well-known notion from the Sturm-Liouville theory
 with complex-valued potentials is one of the main findings of this paper.

Moreover, in this paper, we then develop for the non-Hermitian
Hamiltonian \eqref{HAM1} a spectral theory
which takes as a guiding principle the items (b)--(d) from above.
For simplicity, we restrict ourselves to the physically relevant limit-point case I or,
what is the same, to the case when  $\Gamma$ lies in two Sokes wedges
(see \cite{AT,AT12} for some investigations in the limit-circle case).

Similar as in item (b) from above, we characterize the domains of the minimal operator $A_{0\pm}(\tau_{\pm})$ and the maximal operator $A_{max\pm}(\tau_{\pm})$ as
$$\mathrm{dom}\,A_{max\pm}(\tau_{\pm}):=\{w \in
L^2(\mathbb{R}_{\pm}): \tau_{\pm} w \in L^2(\mathbb{R}_{\pm}),
 w, w' \in AC_{loc}(\mathbb{R}_{\pm})\}
$$
and
$$\mathrm{dom}\,A_{0\pm}(\tau_{\pm}):=\{w \in
\mathrm{dom}\,A_{max\pm}(\tau_{\pm}): w(0)=w'(0)=0\}
$$
(in the limit-point case I).
The minimal operator is now $\mathcal{T}$-symmetric (in the literatur $J$-symmetric, that is, symmetric under complex conjugation, see also Section \ref{sec:2} below)
 and its adjoint is the maximal operator, i.e.\ we show
$$
A_{max\pm}(\tau_{\pm})^*=A_{0\pm}(\overline{\tau}_{\pm}).
$$
The maximal/minimal operators $A_{max+}(\tau_{+})$
and $A_{0+}(\tau_{+})$ corresponds to the differential expression
$\tau_+$ on the positive real axis, cf.\ \eqref{SLEE},
whereas $A_{max-}(\tau_{-})$
and $A_{0-}(\tau_{-})$ correspond to $\tau_-$ on
 $\mathbb R_-$, cf.\ \eqref{SLEE1}. However, the problem
 under consideration  is  \eqref{HAMEQ}, which corresponds
 (after parametrization) to the joint problems \eqref{SLEE}
 and \eqref{SLEE1} on the real line with a (so far)
 unspecified boundary condition in zero.

 Hence, we will use the maximal/minimal operators $A_{max\pm}(\tau_{\pm})$ and $A_{0\pm}(\tau_{\pm})$ as the building blocks for
 operators on the full axis.
We define the maximal operator on the full-axis via the direct sum of the maximal operators on the half-axis,
$$
A_{max}=A_{max-}(\tau_{-})\oplus A_{max+}(\tau_{+})
$$
and domain
$$
D_{max}=\left\{w \in L^2(\mathbb{R}):  Aw \in L^2(\mathbb{R}), w \vert _{\mathbb{R_\pm}}, w'\vert _{\mathbb{R_\pm}}  \in AC_{loc}(\mathbb{R_\pm})\right\}.
$$
 Moreover we obtain in the same way the minimal operator
$$
A_0=A_{0-}(\tau_{-})\oplus A_{0+}(\tau_{+})
$$
with domain
$$
\mathrm{dom}\, A_0=\{w\in D_{max}: w(0+)=w(0-)=w'(0+)=w'(0-)=0\}.
$$
It turns out that the operators $A_{max}$ and $A_0$ are adjoint to each other in the new inner product $[\cdot,\cdot]$,
see, e.g., \cite{M05,M06,M10,T06}, where
$[\cdot,\cdot]$ is a new inner product defined via
$$
[\cdot,\cdot]:=(\mathcal{P}\cdot,\cdot).
$$
Here $(\cdot,\cdot)$ stands for the classical $L^2$-inner product.
However, when it comes to the spectrum, both operators, the
maximal $A_{max}$ and the minimal $A_{0}$, are not suitable. Therefore, it is
natural to assume some coupling in zero of the half-axis
operators. This is done by boundary conditions in zero.
From the physical point of view we always assume continuity
in zero, whereas we allow some freedom for the derivative
in zero. Therefore we introduce a parameter
$\alpha$. Finally, we obtain the wanted operator $A$,
\begin{align*}
\mathrm{dom}\,(A):=\left\{w \in D_{max}: w(0+)=w(0-), w'(0+)=\alpha w'(0-) \right\}
\end{align*}
and
\begin{align*}
Aw(x):=\left\{
\begin{array}{cc}
-e^{-2i\phi}w''(x)-(ix)^{N+2}e^{(N+2)i\phi}w(x), & x \geq 0\\
-e^{2i\phi}w''(x)-(ix)^{N+2}e^{-(N+2)i\phi}w(x), & x \leq 0
\end{array}\right.
\end{align*}
We show that the operator $A$ is indeed $\mathcal{PT}$-symmetric and even self-adjoint in the new inner product $[\cdot,\cdot],$ for the right choice of $\alpha$.

In a next step, it is our aim to discuss the spectrum of $A$. For non-self-adjoint operators like $A$ there
is no standard theory to do this. Therefore we use a different
extension of the minimal operator $A_0$ as an aid. For this
we introduce the operator $A_{\pm}$ which are extensions
of the half-axis minimal operators
(or, what is the same, restrictions of the
half-axis maximal operators) with domain
$$
\mathrm{dom}\, A_{\pm}:=\{w \in \mathrm{dom}\,A_{max\pm}(\tau_{\pm}): w(0)=0\}.
$$
 From \cite{BCEP} it is known that the operators $A_{\pm}$ are $\mathcal{T}$-self-adjoint and their spectra consist only of isolated eigenvalues with finite algebraic multiplicity and empty essential spectrum.

Obviously $A$ and  the direct sum of $A_-\oplus A_+$ differ
only by two dimensions.
As a second main result of this note we show that $A$ has the same spectral properties as the direct sum $A_-\oplus A_+$, i.e.\
the spectrum $\sigma(A)$ of $A$ consists only of isolated eigenvalues with finite algebraic multiplicity,
that is, $\sigma(A)=\sigma_p(A)$, the essential spectrum is empty and the resolvent set $\rho(A)$ is non-empty.

Summing up, to some extend it is a surprise that in the physical literature, starting from the seminal paper of C.M.~Bender and S.~Boettcher \cite{BB98}, the above presented techniques from the Sturm-Liouville theory for complex potentials were never exploited. It is the aim of this paper to recall those techniques and, hence, provide a setting of the (nowadays) classical Bender-Boettcher-theory in terms of the spectral extension theory for Sturm-Liouville expressions with a complex potential.

\section{Limit-point/limit-circle and Stokes wedges and lines}
\label{sec:1}
We consider the Hamiltonian
\begin{align*}
H=\frac{1}{2m}p^2-(iz)^{N+2}, \quad z\in \Gamma,
\end{align*}
with a natural number $N>0$, cf.~\cite{B07,BB98} and a wedge-shaped contour, $$\Gamma:=\left\{z=xe^{i \phi sgn(x)}: x \in \mathbb{R}\right\}$$ for some angle $\phi \in (-\pi/2,\pi/2)$, see also \cite{AT14}. We refer to \cite{BT16,M05,M16} where a similiar contour was used. The associated Schr\"odinger eigenvalue problem is
\begin{align}\label{SEE}
-y''(z)-(iz)^{N+2}y(z)=\lambda y(z), \quad z \in \Gamma,
\end{align}
for some complex number $\lambda$.
We map the problem back to the real line via the parametrization
\begin{align}\label{param}
z:\mathbb{R} \rightarrow \mathbb{C}, \quad
z(x):=xe^{i \phi sgn(x)}.
\end{align}
Thus $y$ solves \eqref{SEE} if and only if $w$, $w(x):=y(z(x))$, solves
\begin{align}\label{EVE2}
-e^{\mp 2i \phi}w''(x)-(ix)^{N+2}e^{\pm(N+2)i\phi}w(x)=\lambda w(x), \, x \in \mathbb{R}_{\pm}.
\end{align}
Here and in the following we set $\mathbb{R}_+:=[0,\infty)$ and $\mathbb{R}_-:=(-\infty,0]$. 
For a complex number $z$ with argument $\theta \in (-\pi,\pi]$, we choose as the $n$-th root $z^{1/n}=r^{1/n}e^{i\theta /n}$. In the following theorem we give a classification of this equation into two cases, namely limit-point case and limit-circle case.
\begin{theorem}\label{theorem}
For all $\lambda \in \mathbb{C}$, exactly one of the following holds.
\begin{itemize}
\item[\rm (I)] If $\phi \neq -\frac{N+2}{2N+8}\pi+\frac{2k}{4+N}\pi$, $k=0,\ldots ,N+3$, there exists a, up to a constant, unique solution $w$ of \eqref{EVE2} satisfying $w \in L^2(\mathbb{R}_{\pm})$. In particular there is one solution of \eqref{EVE2} which is not in $L^2(\mathbb{R}_{\pm})$.
\item[\rm (II)] If $\phi = -\frac{N+2}{2N+8}\pi+\frac{2k}{4+N}\pi$, $k=0,\ldots ,N+3$, all solutions $w$ of \eqref{EVE2} satisfy $w \in L^2(\mathbb{R}_{\pm})$.
\end{itemize}
Case {\rm (I)} is called \emph{limit-point} case I and case
{\rm (II)} is called \emph{limit-circle} case.
\end{theorem}
\begin{proof}
We consider equation \eqref{EVE2} on $\mathbb{R}_+$ only. The result for $\mathbb{R}_-$ are obtained by an analogous argument by replacing $x$ by $-x$. This theorem is a special case of \cite[Theorem 2.1]{BCEP}. The two corresponding linear independent solutions $w_1$ and $w_2$ of the Schr\"odinger eigenvalue differential equation $-w''(x)-(ix)^{N+2}e^{(N+4)i\phi}w(x)=\tilde \lambda w(x), \, x \in \mathbb{R}_{+}$, $\tilde \lambda=e^{2i\phi}\lambda$, satisfy \cite[Corollary 2.2.1]{E}
\begin{align}\label{asymptotic}
w_{1,2}(x) \sim q(x)^{-1/4}\mathrm{exp}\left(\pm \int_1^x \mathrm{Re} (q(t)^{1/2})\, dt\right) \text{ for } x \rightarrow \infty
\end{align}
with $q(x):=-(ix)^{N+2}e^{(N+4)i\phi}-\lambda e^{ 2i \phi}$. The notation $f(x) \sim g(x)$ means that $f(x)/g(x) \rightarrow 1$ as $x \rightarrow \infty$.

We compute $\mathrm{Re} (q(t)^{1/2}).$ For $\lambda=0$ we obtain
\begin{align*}
\mathrm{Re} (q(t)^{1/2})& =\mathrm{Re} ((-(ix)^{N+2}e^{(N+4)i\phi})^{1/2})\\ &= \mathrm{Re} ((e^{i\pi +(N+2)i\pi /2+ (N+4)i\phi})^{1/2})x^{(N+2)/2}\\ &= \mathrm{Re} (e^{i\pi /2 +(N+2)i\pi /4+ (N+4)i\phi/2})x^{(N+2)/2}\\ &= -\sin((N+2)\pi /4+ (N+4)\phi/2)x^{(N+2)/2}
\end{align*}
It is easy to see that
\begin{align*}
\sin ((N+2)\pi /4+ (N+4)\phi/2)=0
\end{align*}
if and only if
\begin{align*}
\phi = -\frac{N+2}{2N+8}\pi+\frac{2k}{4+N}\pi, \text{ for } k\in \mathbb{Z}.
\end{align*}
Hence, if $\phi \neq -\frac{N+2}{2N+8}\pi+\frac{2k}{4+N}\pi$ and if $\lambda=0$ then $\mathrm{Re} (q(t)^{1/2}) \neq 0$ and there exists exactly one solution  in $L^2(\mathbb{R}_+)$ or $L^2(\mathbb{R}_-)$, respectively. This implies, see \cite[Theorem 2.1]{BCEP}, that we have case (I), limit-point case I for $\lambda=0$ and with \cite[Remark 2.2]{BCEP} even for all $\lambda \in \mathbb{C}$. This shows (I).

It remains to consider the case $\phi=-\frac{N+2}{2N+8}\pi+\frac{2k}{4+N}\pi$ and $k\in \mathbb{Z}.$ We obtain $$q(x)=-(ix)^{N+2}e^{-(N+2)i\pi/2 +2ki\pi}-\tilde \lambda=-x^{N+2}-\tilde \lambda$$ and the Schr\"odinger eigenvalue equation $$-w''(x)-x^{N+2}w(x)=\tilde \lambda w(x)$$ and we know from \eqref{asymptotic} that both (linearly independent) solutions of \eqref{EVE2} are in $L^2(\mathbb{R}_+)$, because for $\tilde \lambda=0$ we obtain $\mathrm{Re} (q(t)^{1/2})=0$. Therefore from \cite[Theorem 2.1]{BCEP} we have to examine whether
\begin{align}\label{BCEPTHM}
\int_0^{\infty} \mathrm{Re}\, e^{i\eta}\left(|w'(x)|^2+(-x^{N+2}-K)|w(x)|^2\right)\, dx +\int_0^{\infty}|w(x)|^2\, dx <\infty
\end{align}
is for one or both solutions of \eqref{EVE2} fulfilled, where $\eta$ und $K$ are suitable variables, which we explain in the following, in order to decide wether we are in the limit-point case I, II or limit-circle case. In our setting the set $$Q_+:=\mathrm{clconv}\left\{r-x^{N+2}:x\in [0,\infty), 0<r<\infty\right\},$$ where $\mathrm{clconv}$ denotes the closed convex hull, is the real line and $K$ is the number in $Q_+$ with the shortest distance to $\lambda$, hence $K =\mathrm{Re}\, \lambda.$ And $\eta$ corresponds to the angle which rotates $Q_+$ into the right (closed) half plane, such that $\lambda$ is located in the left half plane, hence $\eta=\pm \frac{\pi}{2}.$ So $$\pm\int_0^{\infty} \mathrm{Re}\, i\left(|w'(x)|^2+(-x^{N+2}-\mathrm{Re}\, \lambda)|w(x)|^2\right)\, dx=0.$$ Condition \eqref{BCEPTHM} is fulfilled for both solutions. Thus we are in the limit-circle case (i.e.\ case III in \cite{BCEP}). \qed
\end{proof}

\begin{remark}
In particular limit-point case II (cf.\ Section \ref{intro})
is not possible, which corresponds to case (II) in \cite[Theorem 2.1]{BCEP}.
\end{remark}

\begin{remark}
The limit-point case I, II and limit-circle case correspond to the cases I, II and III from \cite{SIMS} and \cite{BCEP}.
\end{remark}

In the limit-point case there is exactly one solution of \eqref{EVE2} which is in $L^2(\mathbb{R}_+)$ resp.\ $L^2(\mathbb{R}_-)$ and because of the asymptotics \eqref{asymptotic} we even know that this solution goes exponentially to $0$ for $|x| \rightarrow \infty$.
 The regions in the complex plane where
$\Gamma$ fulfills this condition are wedges, see
e.g.\ \cite{BB98,M05,M10}.

We decompose the complex plane according to the angle $\theta = -\frac{N+2}{2N+8}\pi+\frac{2k}{4+N}\pi$ in $N+4$ sectors
\begin{align*}
S_k :&=\left\{z \in \mathbb{C}: -\frac{N+2}{2N+8}\pi+\frac{2k-2}{4+N}\pi<\mathrm{arg}(z)<-\frac{N+2}{2N+8}\pi+\frac{2k}{4+N}\pi \right\},\\ k &= 0,\ldots, N+3.
\end{align*}
The boundary of each $S_k$ consists of two rays $L_k$
\begin{align*}
L_k:=\left\{z \in \mathbb{C}: \mathrm{arg}(z)=-\frac{N+2}{2N+8}\pi+\frac{2k}{4+N}\pi \right\}, \, k=0,\ldots, N+3.
\end{align*}
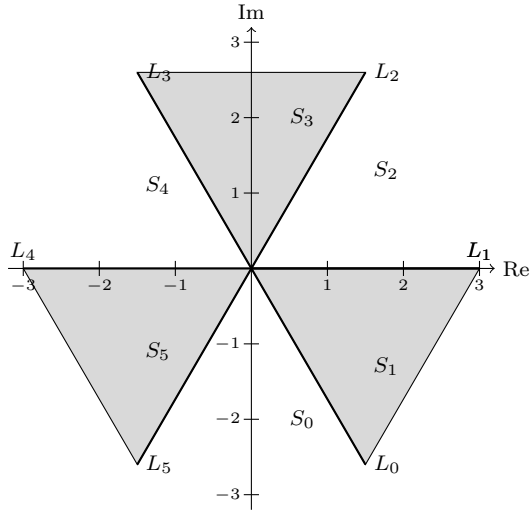
\begin{figure}
\begin{minipage}[h]{\textwidth}
\begin{center}
\begin{tikzpicture}[domain=-3:3]\label{Figg2}
\filldraw[fill opacity=0.15] (0,0) -- (1.5, 2.6) -- (-1.5, 2.6) --cycle;
\filldraw[fill opacity=0.15] (0,0) -- (-3 ,0) -- (-1.5, -2.6) --cycle;
\filldraw[fill opacity=0.15] (0,0) -- (3 ,0) -- ( 1.5, -2.6) --cycle;
\draw[->] (-3.2,0) -- (3.2,0) node[right] {$\mathrm{Re}$}; \draw[->] (0,-3.2)
-- (0,3.2) node[above] {$\mathrm{Im}$};
\foreach \x in {-3,-2,-1,1,2,3} \draw (\x,-.1) -- (\x,.1)
node[below=4pt] {$\scriptstyle \x$}; \foreach \y in
{-3,-2,-1,1,2,3} \draw (-.1,\y) -- (.1,\y) node[left=4pt]
{$\scriptstyle \y$}; \draw[-,thick] (0,0) -- ( -1.5, 2.6)
node[right] {$L_3$};  \draw[-,thick] (0, 0) -- (-3 ,0)node[above] {$L_4$};
\draw[-,thick] (0,0) -- ( -1.5, -2.6)
node[right] {$L_5$};\draw[-,thick] (0,0) -- ( 1.5, 2.6)
node[right] {$L_2$};  \draw[-,thick] (0, 0) -- (3 ,0)node[above] {$L_1$};
\draw[-,thick] (0,0) -- ( 1.5, -2.6)
node[right] {$L_0$};
\draw[-,thick] (0, 0) -- (3 ,0)node[above] {$L_1$};

\draw ( 1.5, -1.3) node[right] {$S_1$};
\draw ( 1.5, 1.3) node[right] {$S_2$};
\draw ( 0.4, 2.0) node[right] {$S_3$};
\draw ( -1.5, 1.1) node[right] {$S_4$};
\draw ( -1.5, -1.1) node[right] {$S_5$};
\draw ( 0.4, -2) node[right] {$S_0$};
\end{tikzpicture}\\
\caption{Stokes lines $L_k$ and Stokes wedges $S_k$ for $N=2$}
\end{center}
\end{minipage}
\end{figure}

In the sectors $S_k,$ $k=0,\ldots, N+3$ one solution of \eqref{EVE2} decays exponentially, wheras on the lines $L_k$ both solutions decay polynomially. The regions $S_k$ are called \emph{Stokes wedges} $S_k$ (see i.e.\ \cite{B07,BB98,BBCJMW06}) and the rays $L_k$ are called \emph{Stokes lines}. Hence we have $N+4$ Stokes lines and Stokes wedges.

By definition, $\Gamma$ is either contained in two Stokes wedges or corresponds to two Stokes lines. This means we can classify our problem depending on the angle $\phi$ of the contour $\Gamma$.

\begin{theorem}\label{prop}
\begin{itemize}
\item[\rm (i)] If $\Gamma$ is located in two Stokes wedges, which are symmetric with respect to the imaginary axis, then \eqref{EVE2} is in the limit-point case for all $\lambda \in \mathbb{C}$, cf.~case (I) in Theorem \ref{theorem}. In particular this implies that only one solution of \eqref{EVE2} is in $L^2(\mathbb{R_+})$ resp.\ $L^2(\mathbb{R_-})$.
\item[\rm (ii)] If $\Gamma$ is located in on Stokes lines, then \eqref{EVE2} is in the limit-circle case for all $\lambda \in \mathbb{C}$, cf.~case (III) in Theorem \ref{theorem}. In particular this implies that all solutions of \eqref{EVE2} are in $L^2(\mathbb{R_+})$ resp.\ $L^2(\mathbb{R_-})$.
\end{itemize}
\end{theorem}

\section{Maximal and minimal operators on the semi-axis}
\label{sec:2}

From now on we restrict ourselves to the limit-point case,
i.e.\ $\Gamma$ lies in two Stokes wedges and \eqref{EVE2} has exactly one solution which is in $L^2(\mathbb{R}_{\pm})$,
cf.\ Theorem \ref{prop}.
Here we will define three different kinds of operators on
$\mathbb R_+$ and $\mathbb R_-$: The maximal, the minimal
and the preminimal operator. This is motivated by the classical
procedure for Sturm-Liouville expressions in the
limit-point case. In the classical Sturm-Liouville
situation, where the coefficients are real, the minimal
operator is the closure of the preminimal, it is a symmetric
operator in a Hilbert space and its adjoint is the maximal
operator.

Here, the situation is slightly different.  However, the definitions of the corresponding  operators are formally
the same as in the
classical Sturm-Liouville case but due to the complex-valued coefficients the adjoints behave differently.

\begin{definition}
The operator $\mathcal{T}$ defined on the Hilbert space $L^2(I)$, where $I\subset \mathbb{R}$ is an interval, is called \emph{time reverse} operator, if for all $u \in L^2$ we have
$$
\mathcal{T}u(x)=\overline{u}(x).
$$
\end{definition}
We mention that in \cite{EE} $\mathcal{T}$ equals $J$.

We consider the following differential expressions
\begin{align*}
\tau_{\pm} w(x):=-e^{\mp 2i\phi}w''(x)-(ix)^{N+2}e^{\pm(N+2)i\phi}w(x)
\end{align*}
and the formal adjoint
\begin{align}\label{tau_adjungierte}
\tau^+_{\pm} w(x):=-e^{\pm 2i\phi}w''(x)-(-ix)^{N+2}e^{\mp (N+2)i\phi}w(x)
\end{align}
on $\mathbb{R}_{\pm}.$ Obviously
\begin{align}\label{tau_adjungierte22}
\tau^+_{\pm}=\overline{\tau_{\pm}}, \quad \mbox{where } \overline{\tau_{\pm}}=\mathcal{T}\tau_{\pm}\mathcal{T}.
\end{align}

We assume that $\tau_{\pm}$ is in the limit-point case, that is, $\phi \neq -\frac{N+2}{2N+8}\pi+\frac{2k}{N+4}\pi$, cf.\ Theorem \ref{prop}. Observe that then also the following lemma holds.
\begin{lemma}\label{lemma_lpc}
If $\tau_{\pm}$ is in the limit-point case, then $\tau^+_{\pm}=\mathcal{T}\tau_{\pm}\mathcal{T}$ is in the limit-point case.
\end{lemma}
\begin{proof}
As in the proof of Theorem \ref{theorem} we use the asymptotics \eqref{asymptotic} from \cite[Corollary 2.2.1]{E} and calculate for the potential in \eqref{tau_adjungierte} its real part for $x\in \mathbb{R}_+$
\begin{align*}
\mathrm{Re} (q(t)^{1/2})& =\mathrm{Re} ((-(-ix)^{N+2}e^{-(N+4)i\phi})^{1/2})\\ &= \mathrm{Re} ((e^{i\pi -(N+2)i\pi /2- (N+4)i\phi})^{1/2})x^{(N+2)/2}\\ &= \mathrm{Re} (e^{i\pi /2 -(N+2)i\pi /4- (N+4)i\phi/2})x^{(N+2)/2}\\ &= \sin((N+2)\pi /4+ (N+4)\phi/2)x^{(N+2)/2}.
\end{align*}
Hence $$\sin((N+2)\pi /4+ (N+4)\phi/2)x^{(N+2)/2}\neq 0,$$ if and only if $$\phi \neq -\frac{N+2}{2N+8}\pi+\frac{2k}{N+4}\pi, \quad k\in \mathbb{Z},$$ which is exactly the condition for
 $\tau_{+}$ to be in the limit-point case, see Theorems
\ref{theorem} and \ref{prop}. In the same way we obtain the result for $x\in \mathbb{R}_-$.
\qed
\end{proof}

Define the following operators with
\begin{align*}
\mathrm{dom}\,(A'_{0\pm}(\tau_{\pm})):=\left\{w \in L^2(\mathbb{R}_{\pm}): \tau_{\pm} w \in L^2(\mathbb{R}_{\pm}), w, w' \in AC_{loc}(\mathbb{R}_{\pm}),\right.\\ \left. w(0)=w'(0)=0, w \text{ has compact support in } \mathbb{R}_{\pm} \right\}
\end{align*}
\begin{align*}
A'_{0\pm}(\tau_{\pm})w(x):=\tau_{\pm}w(x).
\end{align*}
By $A_{0\pm}(\tau_{\pm})$ we denote the closure of $A'_{0\pm}(\tau_{\pm})$ ($A'_{0\pm}(\tau_{\pm})$ is closable by \cite[III Theorem 10.7]{EE}). The operators $A'_{0\pm}(\tau_{\pm})$
correspond to the preminimal operators in classical Sturm-Liouville theory, whereas $A_{0\pm}(\tau_{\pm})$ correspond
to the minimal operators.

Additionally we define the maximal operators
\begin{align*}
\mathrm{dom}\,(A_{max\pm}(\tau_{\pm})):=\left\{w \in L^2(\mathbb{R}_{\pm}): \tau_{\pm} w \in L^2(\mathbb{R}_{\pm}), w, w' \in AC_{loc}(\mathbb{R}_{\pm})\right\}
\end{align*}
\begin{align*}
A_{max\pm}(\tau_{\pm})w(x):=\tau_{\pm}w(x).
\end{align*}

Recall that for a closed operator $T:\mathrm{dom}\, (T)\subset L^2\rightarrow L^2$ the \emph{deficiency} of $T$ is defined as $\mathrm{def}\, T:=\mathrm{dim}\, L^2 /\mathrm{ran}\, (T)$.
Moreover, we recall that the notion of the set $\Pi(T)$ of \emph{regular points of} $T$  (cf., e.g., \cite[pg.\ 101]{EE}) is
$$
\Pi(T):=\left\{\lambda:\exists \, k(\lambda)>0 \text{ with }\|(A-\lambda)u\| \geq k(\lambda)\|u\|\text{ for all } u\in \mathrm{dom}\,(A)\right\}.
$$

\begin{theorem}\label{Hilfsstory}
We have
\begin{equation}\label{regular}
A_{max\pm}(\tau_{\pm})^*=A_{0\pm}(\tau^+_{\pm}) \text{ and }A_{0\pm}(\tau_{\pm})^*=A_{max\pm}(\tau^+_{\pm}).
\end{equation}
Moreover $\mathcal{T}A_{0\pm}(\tau_{\pm})\mathcal{T} \subset A_{0\pm}(\tau_{\pm})^*$ and $\mathrm{def}\, (A_{0\pm}(\tau_{\pm})-\lambda)=\mathrm{def}\, (A_{0\pm}(\tau^+_{\pm})-\overline{\lambda})$ is either $1$ or $2$ for all $\lambda \in \Pi (A_{0\pm}(\tau_{\pm}))$. In the limit-point case we obtain $\mathrm{def}\, (A_{0\pm}(\tau_{\pm})-\lambda)=1$ and
\begin{equation}\label{d.d}
\mathrm{dim}\, \mathrm{dom}\,A_{max\pm}(\tau_{\pm}) /\mathrm{dom}\,A_{0\pm}(\tau_{\pm})=2.
\end{equation}
Furthermore, in the limit-point case, $\Pi (A_{0\pm}(\tau_{\pm}))\neq \emptyset$ and with
\begin{align*}
Q_{\pm}:=\mathrm{clconv}\, \left\{e^{\mp 2i\phi}r-(ix)^{N+2}e^{\pm (N+2)i\phi}: 0<r<\infty, x\in \mathbb{R}_{\pm}\right\}
\end{align*}
we have
\begin{align}\label{eq_numwerte}
\mathbb{C}\backslash Q_{\pm} \subset \Pi (A_{0\pm}(\tau_{\pm})).
\end{align}
In particular, $Q_+$ and $Q_-$ are sectors in the complex plane
with opening angles strictly less than $\pi$,
\begin{equation}\label{13a}
\Pi (A_{0-}(\tau_-)) \cap \Pi (A_{0+}(\tau_+)) \neq
\emptyset.
\end{equation}
\end{theorem}

\begin{proof}
We will use  \cite[III Theorem 10.7]{EE}. It cannot
be used directly as the coefficient in front
of the second derivative in \cite[III Theorem 10.7]{EE}
is assumed to be real-valued. However, a multiplication
in \eqref{EVE2} by $e^{\pm 2i\phi}$ turns the eigenvalue
problem \eqref{EVE2}  into a problem considered
in  \cite[III Section 10]{EE} (with a shifted eigenvalue
parameter). Then \cite[III Theorem 10.7]{EE} holds for the
shifted problem and, again by a multiplication with
$e^{\mp 2i\phi}$, we see that  \cite[III Theorem 10.7]{EE}
is also valid for \eqref{EVE2}. Therefore it remains only to show \eqref{eq_numwerte} and that in the limit-point case $\mathrm{def}\, (A_{0\pm}(\tau_{\pm})-\lambda)=1$ and \eqref{d.d} hold.

Observe that
$$
Q_-^*:=\left\{\overline{x}:x\in Q_- \right\}=Q_+
$$
and $Q_{\pm}$ are convex sectors in the complex plane.
Assume that their opening is $\pi$, then we have for $x \in \mathbb{R}_+$ and some $k\in \mathbb Z$
$$
-2 \phi+2k\pi = \frac{\pi}{2}(N+2) +(N+2)\phi,
$$
and this gives
$$
\phi= \frac{2k\pi}{N+4} -\frac{(N+2)\pi}{2N+8}.
$$
For $x \in \mathbb{R}_-$ we obtain the same condition as
$Q_-^*=Q_+$. But this condition
is the condition for the limit-circle case and hence not possible, see Theorems \ref{theorem} and \ref{prop}. Therefore, the opening angle of $Q_\pm$ is strictly less then $\pi$ and  we have
\begin{equation}\label{13b}
Q_{+}\cup Q_- \neq \mathbb{C}.
\end{equation}
We  choose $\lambda \in \mathbb{C}\backslash Q_{\pm}$.
Because $Q_{\pm}$ are sectors with two rays as boundary (which may coincide) the distance $\delta(\lambda)$ between $\lambda$ and $Q_{\pm}$ is $\delta(\lambda)=|K-\lambda|$, where
$K$ is a point of the boundary of $Q_\pm$, i.e.,
$K\in \left\{e^{\mp 2i\phi}r: 0<r<\infty \right\}$ or $K\in R_{\pm}:=\left\{-(ix)^{N+2}e^{\pm (N+2)i\phi}: x\in \mathbb{R}_{\pm}\right\}$, cf.\ Figure \ref{Fig3}. There is a suitable angle $\eta \in (-\pi,\pi]$ with
$$
\delta(\lambda)=|K-\lambda|= e^{i\eta} (K-\lambda)
$$
The convexity of $Q_{\pm}$ induce that the straight line
$$
\left\{e^{\mp 2i\phi}r: r\in \mathbb{R} \right\}\text{ or resp.\ }\left\{-(i)^{N+2}e^{\pm i(N+2)\phi}s: s\in \mathbb{R}\right\}
$$
seperates $\lambda$ and $Q_{\pm}$, cf.\ Figure \ref{Fig3}.  Moreover we get after a rotation via the angle $\eta$ that $Q_{\pm}$ is located in the right half plane, cf.\ Figure \ref{Fig4},
\begin{align*}
\{z=e^{i\eta}q:q\in Q_{\pm}\} \subset \mathbb{C}_{\mathrm{Re}\, \geq 0}.
\end{align*}
\begin{figure}[ht]
\begin{minipage}[h]{0.48\textwidth}
\begin{center}
\begin{tikzpicture}[domain=-2:3]\label{Fig3}
\filldraw[fill opacity=0.15] (0,0) -- (2.2, 3.0) -- (-1.5, 3.0) --cycle;
\draw[->] (-2.2,0) -- (3.2,0) node[right] {$\mathrm{Re}$}; \draw[->] (0,-2.2)
-- (0,3.2) node[above] {$\mathrm{Im}$};
\foreach \x in {-2,-1,1,2,3} \draw (\x,-.1) -- (\x,.1)
node[below=4pt] {$\scriptstyle \x$}; \foreach \y in
{-2,-1,1,2,3} \draw (-.1,\y) -- (.1,\y) node[left=4pt]
{$\scriptstyle \y$}; \draw[-,thick] (-1.5,3.0) -- (1,-2.0)
node[right] {$\{e^{-2i\phi}r:r\in \mathbb{R}\}$};
\coordinate[label=left:$\lambda$] (A) at (-2,0.5);
\coordinate[label=left:$K$] (B) at (-0.7,1.4);
\draw (A) -- (B);
\fill (A) circle (2pt);
\fill (B) circle (2pt);
\draw ( 0.4, 2.0) node[right] {$Q_+$};
\draw ( 2, 2.3) node[right] {$R_+$};

\end{tikzpicture}\\
\caption{The line $\{e^{-2i\phi}r:r\in \mathbb{R}\}$ seperates $\lambda$ and $Q_+$}
\end{center}
\end{minipage} \ \ \
\begin{minipage}[h]{0.48\textwidth}
\begin{center}
\begin{tikzpicture}[domain=-2:3] \label{Fig4}
\filldraw[fill opacity=0.15] (0,0)  -- (0, 3.0)-- (3.0,3.0) -- (3.0, 1.5) --cycle;
\draw[->] (-2.2,0) -- (3.2,0) node[right] {$\mathrm{Re}$}; \draw[->] (0,-2.2)
-- (0,3.2) node[above] {$\mathrm{Im}$};
\foreach \x in {-2,-1,1,2,3} \draw (\x,-.1) -- (\x,.1)
node[below=4pt] {$\scriptstyle \x$}; \foreach \y in
{-2,-1,1,2,3} \draw (-.1,\y) -- (.1,\y) node[left=4pt]
{$\scriptstyle \y$}; 
\coordinate[label=left:$e^{i\eta}\lambda$] (A) at (-1.6,1.3);
\coordinate[label=right:$e^{i\eta}K$] (B) at (0,1.3);
\draw (A) -- (B);
\fill (A) circle (2pt);
\fill (B) circle (2pt);
\draw ( 0.4, 2.0) node[right] {$e^{i\eta}Q_+$};
\draw ( 2.2, 1) node[below] {$e^{i\eta}R_+$};
\end{tikzpicture}\\
\caption{Rotation via the angle $\eta$ and $e^{i\eta}Q_+$ lies in the right half plane}
\end{center}
\end{minipage}
\end{figure}
We obtain
\begin{align}\label{eq_intnumerwerte}
\mathrm{Re} \, e^{i\eta}\left(e^{\mp 2i\phi}r-(ix)^{N+2}e^{\pm (N+2)i\phi}-K\right)\geq 0, \quad \text{for } 0<r<\infty, x\in \mathbb{R}_{\pm}
\end{align}
For $\lambda \in \mathbb{C}\backslash Q_{\pm}$ we get for $u \in \mathrm{dom}(A'_{0\pm})$ and $\|u\|=1$
\begin{align*}
&\|(A'_{0\pm}(\tau_{\pm})-\lambda)u\|\geq |(A'_{0\pm}(\tau_{\pm})u,u)-\lambda|\\[1ex]
=& \Big|\int_0^{\infty}-e^{\mp 2i\phi}u''\overline{u}-(ix)^{N+2}e^{\pm (N+2)i\phi}|u|^2\, dx -\lambda\Big|\\[1ex]
=& \Big|\int_0^{\infty}e^{\mp 2i\phi}|u'|^2-(ix)^{N+2}e^{\pm (N+2)i\phi}|u|^2\, dx -\lambda\Big|\\[1ex]
\geq & \mathrm{Re}\, e^{i\eta}\left(\int_0^{\infty}e^{\mp 2i\phi}|u'|^2-(ix)^{N+2}e^{\pm (N+2)i\phi}|u|^2-K|u|^2\, dx +K-\lambda\right)
\end{align*}
Now \eqref{eq_intnumerwerte} implies
\begin{align*}
& \int_0^{\infty}\mathrm{Re}\, e^{i\eta}\left(e^{\mp 2i\phi}|u'|^2-(ix)^{N+2}e^{\pm (N+2)i\phi}|u|^2-K|u|^2\right)\, dx +\mathrm{Re}\, e^{i\eta}(K-\lambda)\\
&\geq  \delta(\lambda)>0.
\end{align*}
Hence $\mathbb{C}\backslash Q_{\pm} \subset \Pi (A'_{0\pm}(\tau_{\pm}))$ and in particular we have $\Pi (A'_{0\pm}(\tau_{\pm}))\neq \emptyset$. Now we choose for $y \in \mathrm{dom}\, A_{0\pm}(\tau_{\pm})$ a sequence $(x_n)_n\subset \mathrm{dom}\, A'_{0\pm}(\tau_{\pm})$ such that $x_n\rightarrow y$ and $A'_{0\pm}(\tau_{\pm}) x_n\rightarrow A_{0\pm}(\tau_{\pm})y$. Moreover for $\varepsilon>0$ choose $n$ large enough, such that $\|A'_{0\pm}(\tau_{\pm}) x_n-A_{0\pm}(\tau_{\pm})y\|\leq \varepsilon$ and $k(\lambda)\|x_n-y\| \leq \varepsilon.$ Then we obtain
\begin{align*}
&\|(A_{0\pm}(\tau_{\pm})-\lambda)y\|=\|(A_{0\pm}(\tau_{\pm})-\lambda)(y-x_n)+(A'_{0\pm}(\tau_{\pm})-\lambda)x_n\|\\[1ex]
&\geq \|(A'_{0\pm}(\tau_{\pm})-\lambda)x_n\|-\varepsilon\geq k(\lambda)\|x_n\| -\varepsilon\geq k(\lambda)\|y\|-2\varepsilon,
\end{align*}
and \eqref{eq_numwerte} follows. Moreover, from this and \eqref{13b} we obtain \eqref{13a}.

Now we can apply \cite[III Theorem 5.6]{EE} and obtain
\begin{align*}
\mathrm{dim}\, \mathrm{dom}\,(\mathcal{T}A^*_{0\pm}(\tau_{\pm})\mathcal{T}) /\mathrm{dom}\,A_{0\pm}(\tau_{\pm})=2\mathrm{def}\, (A_{0\pm}(\tau_{\pm})-\lambda)
\end{align*}
and
\begin{align*}
\mathrm{dom}\,&(\mathcal{T}A_{0\pm}(\tau_{\pm})^*\mathcal{T})\\[1ex]
& =\mathrm{dom}\,A_{0\pm}(\tau_{\pm})\dot{+}\mathrm{ker}\,\left((A_{0\pm}(\tau_{\pm})^*-\overline{\lambda})(\mathcal{T}A_{0\pm}(\tau_{\pm})^*\mathcal{T}-\lambda)\right).
\end{align*}
With
$A^*_{0\pm}(\tau_{\pm})=A_{max\pm}(\tau^+_{\pm})=A_{max\pm}(\overline{\tau}_{\pm})$ and
$$
\mathcal{T}A^*_{0\pm}(\tau_{\pm})\mathcal{T}=\mathcal{T}A_{max\pm}(\overline{\tau}_{\pm})\mathcal{T}=A_{max\pm}(\tau_{\pm})
$$
we obtain
\begin{align*}
\mathrm{dim}\, \mathrm{dom}\,A_{max\pm}(\tau_{\pm}) /\mathrm{dom}\,A_{0\pm}(\tau_{\pm})=2\mathrm{def}\, (A_{0\pm}(\tau_{\pm})-\lambda)
\end{align*}
and
\begin{align}\label{formel_dimker}
&\mathrm{dom}\,(A_{max\pm}(\tau_{\pm}))\\[1ex]\nonumber
& =\mathrm{dom}\,A_{0\pm}(\tau_{\pm})\dot{+}\mathrm{dim}\,\mathrm{ker}\,\left((\mathcal{T}A_{max\pm}(\tau_{\pm})\mathcal{T}-\overline{\lambda})(A_{max\pm}(\tau_{\pm})-\lambda)\right).
\end{align}
Because $\tau_{\pm}$ and $\mathcal{T}\tau_{\pm}\mathcal{T}$ are in the limit-point case, cf.~Lemma \ref{lemma_lpc}, the equations $(\tau_{\pm}-\lambda)u=0$ and $(\mathcal{T}\tau_{\pm}\mathcal{T}-\overline{\lambda})u=0$ have only one solution in $L^2(\mathbb{R}_{\pm})$. Therefore there is only one function $u$ with $(A_{max\pm}(\tau_{\pm})-\lambda)u=0$.
Moreover, we have from \cite[III Theorem 5.6]{EE},
\begin{align*}
\mathrm{dim}\, \mathrm{dom}\,A_{max\pm}(\tau_{\pm}) /\mathrm{dom}\,A_{0\pm}(\tau_{\pm})=2\mathrm{def}\, (A_{0\pm}(\tau_{\pm})-\lambda)
\end{align*}
plus equations \eqref{regular} and \eqref{formel_dimker},
that $\mathrm{dim}\,\mathrm{ker}\,((\mathcal{T}A_{max\pm}(\tau_{\pm})\mathcal{T}-\overline{\lambda})(A_{max\pm}(\tau_{\pm})-\lambda))$ is even and because of the limit-point case at most $2$. Hence
 $$
 \mathrm{dim}\,\mathrm{ker}\,((\mathcal{T}A_{max\pm}(\tau_{\pm})\mathcal{T}-\overline{\lambda})(A_{max\pm}(\tau_{\pm})-\lambda))=2
 $$
and we obtain
\begin{align*}
2=\mathrm{dim}\, \mathrm{dom}\,A_{max\pm}(\tau_{\pm}) /\mathrm{dom}\,A_{0\pm}(\tau_{\pm})=2\mathrm{def}\, (A_{0\pm}(\tau_{\pm})-\lambda).
\end{align*}
\qed
\end{proof}
With \cite[III Theorem 10.13]{EE} the following proposition follows immediately.
\begin{proposition}\label{prop_A0}
We obtain in the limit-point case
$$
\mathrm{dom}\, A_{0\pm}(\tau_{\pm})=\left\{w \in \mathrm{dom}\, A_{0\pm}(\tau_{\pm}): w(0)=w'(0)=0\right\}
$$
and for $u \in \mathrm{dom}\, A_{max\pm}(\tau_{\pm})$
and $v\in \mathrm{dom}\, A_{max\pm}(\tau^+_{\pm})$
$$
\lim_{x\rightarrow \pm \infty}\left(u\overline{v}'-u'\overline{v}\right)(x)=0.
$$
\end{proposition}


\section{Maximal and minimal operators on the full axis}
\label{sec:3}

Here we define and study the maximal and the minimal operator on the real line. We do this by composition of the corresponding operators on the semi-axis from Section \ref{sec:2}.

The maximal operator on $\mathbb{R}$ is given by
$$D_{max}:=\left\{w \in L^2(\mathbb{R}):  \tau_\pm w|_{\mathbb R_\pm} \in L^2(\mathbb{R}), w \vert _{\mathbb{R_\pm}}, w'\vert _{\mathbb{R_\pm}}  \in AC_{loc}(\mathbb{R_\pm})\right\}$$ and
\begin{align*}
A_{max}w(x):=\left\{
\begin{array}{cc}
\tau_+w(x), & x \geq 0,\\
\tau_-w(x), & x \leq 0.
\end{array}\right.
\end{align*}
or, what is the same,
$$
A_{max} = A_{max-}(\tau_{-}) \oplus  A_{max+}(\tau_{+}).
$$

We define the parity $\mathcal P$. One has to be careful
how to define it. In the literature it is quite often just defined by the (somehow sloppy) notion $x\mapsto -x$.
More precisely, we have for a function $f\in L^2(\mathbb{R})$
with $f_+:=f|_{\mathbb R_+}$ and $f_-:=f|_{\mathbb R_-}$
\begin{align*}
(\mathcal{P}f)(x):=\left\{
\begin{array}{ll}
f_-(-x) & \mbox{ if }x\geq 0,\\[1ex]
f_+(-x) & \mbox{ if }x<0.
\end{array} \right.
\end{align*}

The parity $\mathcal{P}$ gives rise to a new inner product,
which was considered in many papers, we mention here only
\cite{M05,M06,M10,T06}.
It is the right inner product in which the operators exhibit symmetry properties, as we will show below,
\begin{align*}
[\cdot,\cdot]=(\mathcal{P}\cdot,\cdot).
\end{align*}
\begin{lemma}\label{lemma_part_int}
For $v,w \in D_{max}$ we have
\begin{align*}
&[A_{max} w,v]-[w,A_{max} v]\\
&=e^{2i\phi}(w'(0+)\overline{v}(0-)+w(0+)\overline{v}'(0-))-e^{-2i\phi}(w'(0-)\overline{v}(0+)+w(0-)\overline{v}'(0+)).
\end{align*}
\end{lemma}
\begin{proof}
As
$\tau^+_{\pm}=\mathcal{T}\tau_{\pm}\mathcal{T}$
(see \eqref{tau_adjungierte22}) we have
$\overline{v}|_{\mathbb R_\pm} \in \mathrm{dom}\,
A_{max\pm}(\tau^+_{\pm})$. From
$$
\tau_\pm \overline{w(-x)} = \overline{\tau_\mp w(-x)},
$$
we see that the function $x\mapsto \overline{w(-x)}$
for $x \in \mathbb R_\pm$, is in $\mathrm{dom}\,
A_{max\pm}(\tau_{\pm})$. Then Proposition \ref{prop_A0} gives
\begin{equation}\label{propeq}
\lim_{x\rightarrow \pm \infty}\overline{w(-x)} v'(x)-\overline{w'(-x)}v(x)=0.
\end{equation}
We have
\begin{align*}
&[A_{max} w,v]-[w,A_{max} v]=(\mathcal{P}A_{max} w,v)-(\mathcal{P}w,A_{max} v)\\
=&\int_{-\infty}^0 \tau_+w(-x) \overline{v}(x)\, dx+\int_0^{\infty} \tau_-w(-x) \overline{v}(x)\, dx\\&-\int_{-\infty}^0 w(-x)\overline{\tau_-v}(x)\, dx-\int_0^{\infty} w(-x)\overline{\tau_+v}(x)\, dx\\
=& \int_0^{\infty}\left(-e^{2i\phi}w''(-x)-(-ix)^{N+2}e^{-(N+2)i\phi}w(-x)\right)\overline{v}(x)\, dx\\&+\int_{-\infty}^{0}\left(-e^{-2i\phi}w''(-x)-(-ix)^{N+2}e^{(N+2)i\phi}w(-x)\right)\overline{v}(x)\, dx\\&-\int_0^{\infty}w(-x)\overline{\left(-e^{-2i\phi}v''(x)-(ix)^{N+2}e^{(N+2)i\phi}v(x)\right)}\, dx\\&-\int_{-\infty}^{0}w(-x)\overline{\left(-e^{2i\phi}v''(x)-(ix)^{N+2}e^{-(N+2)i\phi}v(x)\right)}\, dx\\
=&  e^{2i\phi} \int_0^{\infty}
w(-x)  \overline{v}''(x)-
w''(-x)\overline{v}(x)\, dx\\
&+e^{-2i\phi} \int_{-\infty}^{0}
w(-x)  \overline{v}''(x)-
w''(-x)\overline{v}(x)\, dx.
\end{align*}
Integration by parts gives
\begin{align*}
&[A_{max} w,v]-[w,A_{max} v]\\
=& \lim_{x\rightarrow \infty}e^{2i\phi}(w'(-x)\overline{v}(x)+w(-x)\overline{v}'(x))-\lim_{x\rightarrow -\infty}e^{-2i\phi}(w'(-x)\overline{v}(x)+w(-x)\overline{v}'(x))\\
&+e^{-2i\phi}(w'(0+)\overline{v}(0-)+w(0+)\overline{v}'(0-))-e^{2i\phi}(w'(0-)\overline{v}(0+)+w(0-)\overline{v}'(0+))
\end{align*}
Then \eqref{propeq} (after taking the complex conjugate)
shows the statement of the lemma.\qed
\end{proof}

Similar as the maximal operator on the real line, we define the minimal operator $A_0$ on the real line as the direct sum of the corresponding minimal operators on the half-axis,
\begin{align*}
A_0=A_{0-}(\tau_-)\oplus A_{0+}(\tau_+).
\end{align*}
Observe that with Proposition \ref{prop_A0} the domain of $A_0$ is given via
$$
\mathrm{dom}\, A_0 =
\{w\in D_{max}:w(0+)=w(0-)=w'(0+)=w'(0-)=0\}
$$
and Theorem \ref{Hilfsstory} gives for
$\lambda \in
\Pi(A_{0})=\Pi(A_{0-}(\tau_-)) \cap \Pi(A_{0+}(\tau_+))
$, which is by \eqref{13a} non-empty,
\begin{equation}\label{PlusPunktPlus}
\mathrm{def}\,(A_{0}-\lambda)=\mathrm{def}\,(A_{0-}(\tau_-)-\lambda)+\mathrm{def}\,(A_{0+}(\tau_+)-\lambda)=2.
\end{equation}

Let $A$ be a densely defined, closed operator in $(L^2(\mathbb{R}),\left[\cdot,\cdot\right])$ the adjoint $A^+$ of $A$ with respect to $[\cdot,\cdot]$ is defined on $\mathrm{dom} \, A^+$. This is the set of all $y \in L^2(\mathbb{R})$, such that there is a $z \in L^2(\mathbb{R})$ with
\begin{align*}
[Ax,y]=[x,z],\quad \text{ for all } x \in \mathrm{dom}\, A
\end{align*}
and we set $$A^+y:=z.$$ An operator $A$ is called symmetric with respect to $[\cdot,\cdot]$ (or $[\cdot,\cdot]$-symmetric) if $A \subset A^+$ and self-adjoint with respect to $[\cdot,\cdot]$ (or $[\cdot,\cdot]$-self-adjoint) if $A=A^+$. With Lemma \ref{lemma_part_int} the following follows immediately.

\begin{proposition}\label{Kreinsymmetry}
$A_0$ is symmetric with respect to $[\cdot,\cdot].$ Moreover $A_0^+=A_{max}=A_{max-}(\tau_-)\oplus A_{max+}(\tau_+).$
\end{proposition}

\begin{proof}
It remains to show that $A_0^+=A_{max}.$
With  \eqref{regular} in
Theorem \ref{Hilfsstory} we obtain for $w \in \mathrm{dom}\, A_{max}$ and $v \in \mathrm{dom}\, A_{0}$
\begin{align*}
&[A_{max}w,v]=(\mathcal{P}A_{max}w,v)=(A_{max}w,\mathcal{P}v)\\
&=(A_{max-}(\tau_-)w \vert_{\mathbb{R}_-}, v\vert_{\mathbb{R}_+}(-\cdot))+(A_{max+}(\tau_+)w \vert_{\mathbb{R}_+}, v\vert_{\mathbb{R}_-}(-\cdot))\\
&= (\mathcal{T}A^*_{0-}(\tau_{-})\mathcal{T}w \vert_{\mathbb{R}_-}, v\vert_{\mathbb{R}_+}(-\cdot))+(\mathcal{T}A^*_{0+}(\tau_{+})\mathcal{T}w \vert_{\mathbb{R}_+}, v\vert_{\mathbb{R}_-}(-\cdot))\\ &= (w \vert_{\mathbb{R}_-}, \mathcal{T}A_{0-}(\tau_{-})\mathcal{T}v\vert_{\mathbb{R}_+}(-\cdot))+(w \vert_{\mathbb{R}_+}, \mathcal{T}A_{0+}(\tau_{+})\mathcal{T}v\vert_{\mathbb{R}_-}(-\cdot))\\
&= (w,\mathcal{T}A_0\mathcal{TP}v)=(\mathcal{P}w,\mathcal{PT}A_0\mathcal{TP}v)=(\mathcal{P}w,A_0v)=[w,A_0v],
\end{align*}
because $\mathcal{PT}A_0=A_0\mathcal{PT}.$ So we have $A_0\subset A_{max}^+$ and in a similar way we obtain $A_{max}^+\subset A_{0}$.
\qed
\end{proof}

\begin{remark}
The space $(L^2(\mathbb{R}),\left[\cdot,\cdot\right])$ is a Krein space, see \cite{D95,D99,M05,T06}. For a more
advanced introduction to operators in Krein spaces we
refer to the monographs \cite{AI,Bog}. We mention here only that the operator $A_0$ according to Proposition \ref{Kreinsymmetry} is $[\cdot,\cdot]$-symmetric
in the Krein space $(L^2(\mathbb{R}),\left[\cdot,\cdot\right])$.
\end{remark}

\section{Operator based approach to $\mathcal{PT}$-symmetric Hamiltonians}
\label{sec:4}

In this section we define the operator $A$ corresponding to \eqref{SEE} and \eqref{EVE2} on the full real axis with a coupling condition in $0$. It is an extension of the minimal
operator $A_0$ and a restriction of the maximal
operator $A_{max}$, both studied
in Section \ref{sec:3}.

 Here we restrict ourselves to a coupling of the form
 $w(0+)=w(0-)$ and $w'(0+)=\alpha w'(0-)$ in zero as we want
 $w$, and hence $y$ (see \eqref{SEE}), to be continuous. As we will see below, it is reasonable to allow a jump of $w'$ in $0$. So we define for a fixed complex number $\alpha$ an extension $A$ of $A_0$ by
\begin{align*}
\mathrm{dom}\,(A):=\left\{w \in D_{max}: w(0+)=w(0-), w'(0+)=\alpha w'(0-) \right\}
\end{align*}
\begin{align*}
Au:=A_{max}u.
\end{align*}

\begin{definition}
We call a closed densely defined operator $A$ defined on $L^2(\mathbb{R})$ $\mathcal{PT}$-symmetric if and only if for all $f \in \mathrm{dom}\, A$ we have $\mathcal{PT}f \in \mathrm{dom}\, A$ and $\mathcal{PT}Af=A\mathcal{PT}f$, see also \cite[III. §5.6]{K}.
\end{definition}

\begin{theorem}\label{thm_A_eig}
Let $w\in \mathrm{dom}\, A$ and let $y$ satisfy
$w(x)=y(z(x))$,  where $z$ is given by  \eqref{param}.
Then we have
\begin{itemize}
\item[\rm (i)] $y'$ is continuous if and only if $\alpha=e^{2i\phi}$.
\item[\rm (ii)] $A$ is $\mathcal{PT}$-symmetric if and only if $|\alpha|=1$.
\item[\rm (iii)] $A$ is self-adjoint with respect to $[\cdot,\cdot]$, if and only if $\alpha=e^{-4i\phi}$.
\end{itemize}
\end{theorem}
\begin{proof}
We obtain
$$
w'(x)=z'(x)y'(z(x))=e^{i \phi sgn (x)}y'(z(x)),
$$
 for $x\neq 0$. Then $y'(0+)=y'(0-)$ is equivalent to
 $$
 e^{-i \phi}w'(0+)=y'(0+)=y'(0-)=e^{i \phi}w'(0-).
 $$
  This shows $(i)$.

With $y \in \mathrm{dom}\, A$,
$$
\mathcal{PT}y(0+)=\overline{y(0-)}=\overline{y(0+)}=\mathcal{PT}y(0-)
$$
and
$$
\alpha (\mathcal{PT}y)'(0-)=-\alpha \overline{y'(0+)}=-\alpha\overline{\alpha y'(0-)}=|\alpha|^2(\mathcal{PT}y)'(0+)
$$
we get $\mathcal{PT}y \in \mathrm{dom} \, A$ if and
 only if $|\alpha |=1$. Moreover, for $x>0$ we have
 $$
 \mathcal{PT}Ay(x)=-e^{2i\phi}\overline{y''(-x)}-e^{-(N+2)i\phi }(ix)^{N+2}\overline{y(-x)}=A\mathcal{PT}y(x)
 $$
A similar calculation holds for $x<0$ and (ii) follows.

It remains to show (iii).
From Lemma \ref{lemma_part_int} follows that $A$ is $[\cdot,\cdot]$-symmetric. Because $\mathrm{def}\,(A_{0}-\lambda)=2$ (see \eqref{PlusPunktPlus}) and $A$ is a two-dimensional extension of $A_0$, $A$ is $\left[\cdot,\cdot\right]$-self-adjoint.
\qed
\end{proof}


\begin{proposition}\label{symmetry}
Let $\lambda \in \sigma_p(A)$ and let $|\alpha|=1$,
which implies $\mathcal{PT}$-symmetry for $A$, see Theorem \ref{thm_A_eig}. If $y$ is the corresponding eigenfunction, then $\mathcal{PT}y$ is also an eigenfunction for $\overline{\lambda}$.
\end{proposition}
\begin{proof}
From $y \in \mathrm{dom}\, A$ it follows $\mathcal{PT}y \in \mathrm{dom}\, A$ and $A\mathcal{PT}y=\mathcal{PT}Ay=\mathcal{PT} \lambda y=\overline{\lambda}\mathcal{PT}y$.\qed
\end{proof}

%

The following theorem is our main result.
\begin{theorem}\label{main}
Let  $\alpha=e^{-4i\phi}$.
We assume  $\phi \ne 0$ and we assume that
 one of the following two conditions is satisfied.
 \begin{itemize}
\item If $\phi > 0$, then there exists a natural number $k$, $k\geq 0$, with
$$
\frac{2k\pi}{N+2}-\frac{\pi}{2}<\phi<\frac{(2k+1)\pi}{N+2}-\frac{\pi}{2}.
$$
\item If $\phi< 0$, then there exists $k\in \mathbb{Z}$, $k\leq 0$, with
$$
\frac{(2k-1)\pi}{N+2}-\frac{\pi}{2}<\phi<\frac{2k\pi}{N+2}-\frac{\pi}{2}.
$$
\end{itemize}
Then  $A$ is $[\cdot,\cdot]$-self-adjoint and $\mathcal{PT}$-symmetric with
\begin{align*}
\rho(A) \neq \emptyset, \quad \mbox{and} \quad
 \sigma(A)=\sigma_p(A).
\end{align*}
The spectrum of $A$ is symmetric to the real line, it consists only of discrete eigenvalues of finite algebraic multiplicity with no finite accumulation point and $\mathrm{dim}\, \mathrm{ker}\,(A-\lambda)=1$ for $\lambda \in \sigma_p(A)$.
\end{theorem}
\begin{proof}
The self-adjointness and the $\mathcal{PT}$-symmetry follows
from Theorem \ref{thm_A_eig}. In order to show that the
resolvent set of $A$ is non-empty, we introduce two auxillary
operators $A_\pm$ via
\begin{align*}
\mathrm{dom}\, A_{\pm}
:=\left\{ w\in \mathrm{dom}\, A_{\max\pm}(\tau_\pm)  :  w(0)=0  \right\}, \quad A_{\pm}w(x):=\tau_\pm w(x)
\end{align*}
From \cite[Theorems 4.4 and 4.5]{BCEP} we know, that the spectrum consists at most of isolated eigenvalues with finite algebraic multiplicity and it is located in the set $Q_\pm$,
\begin{align}\label{sigmaPM}
\sigma(A_\pm)=\sigma_p(A_\pm)\subset Q_\pm.
\end{align}
In particular, the essential spectrum is empty.

The assumption on $\phi$ imply that for
 $\phi > 0$ we obtain
  $\sin((N+2)\phi+(N+2)\frac{\pi}{2})>0$ and, hence,
  $\mathrm{Im}\, (-(ix)^{N+2}e^{(N+2)i\phi})<0$.
As $\phi > 0$ is in the interval $(0,\pi/2)$ (see
page \pageref{SEE}), we have $\mathrm{Im}\, e^{-2i\phi}
< 0$ and therefore $Q_+$ is contained in the lower half plane.

If $\phi<0$ we have $\mathrm{Im}\, (-(ix)^{N+2}e^{(N+2)i\phi})>0$ and  $\mathrm{Im}\, e^{-2i\phi}
> 0$ and $Q_+$ is contained in the upper half plane. As $Q_-=Q_+^*$, we obtain
 $$
 Q_+ \cap Q_-=\{0\}.
 $$
\emph{\bf Claim.} For $\lambda \not \in \sigma_p(A_+) \cup \sigma_p(A_-)$  we have $v_{\lambda,+}(0) \neq 0$ and $v_{\lambda,-}(0) \neq 0$, where $v_{\lambda,+}$ and $v_{\lambda,-}$ are the non-zero $L^2$-solutions of $(\tau_{\pm}-\lambda)y=0$. In this case
\begin{align}\label{weylquotient}
\lambda \in \sigma_p(A) \Leftrightarrow
\frac{v'_{\lambda,+}(0)}{v_{\lambda,+}(0)}=
e^{-4i\phi}\frac{v'_{\lambda,-}(0)}{v_{\lambda,-}(0)}.
\end{align}
\begin{proof}\emph{of the claim}.
Suppose that the right hand side of \eqref{weylquotient} holds. Set
\begin{align*}
v(x):=\left\{
\begin{array}{lc}
v_{\lambda,+}(x),& x \geq 0\\[1ex]
\frac{v_{\lambda,+}(0)}{v_{\lambda,-}(0)}v_{\lambda,-}(x), \quad &
x \leq 0
\end{array}\right.
\end{align*}
then $v(0+)=v(0-)$ and
$$
v'(0-)=\frac{v_{\lambda,+}(0)}{v_{\lambda,-}(0)}
v'_{\lambda,-}(0)=e^{4i\phi}v'_{\lambda,+}(0)=
e^{4i\phi}v'(0+).
$$
So we have $v \in \mathrm{dom}\, A$ and $\lambda \in \sigma_p(A)$.

To prove the converse choose an eigenfunction $v \in \mathrm{dom}\,  A$ corresponding to the eigenvalue $\lambda$.
Due to the limit point case
 there exist constants with $v\vert_{\mathbb{R}_{\pm}}=\alpha_{\pm} v_{\lambda,\pm}$. Hence $v(0)=\alpha_+
 v_{\lambda,+}(0)=\alpha_- v_{\lambda,-}(0)$ and
 $\alpha_+ v'_{\lambda,+}(0)=v'(0+)=e^{-4i\phi}v'(0-)=e^{-4i\phi}\alpha_- v'_{\lambda,-}(0)$ and we obtain
\begin{align*}
\frac{v'_{\lambda,+}(0)}{v_{\lambda,+}(0)}=
\frac{\alpha_+v'_{\lambda,+}(0)}{\alpha_+v_{\lambda,+}(0)}
=e^{-4i\phi}\frac{\alpha_-v'_{\lambda,-}(0)}{\alpha_-v_{\lambda,-}(0)}=e^{-4i\phi}\frac{v'_{\lambda,-}(0)}{v_{\lambda,-}(0)}
\end{align*}
\end{proof}
and the claim is proved.

We continue with the proof of Theorem \ref{main}. 
We have $Q_+ \cap Q_-=\{0\}$ and, hence, by
\eqref{sigmaPM} we find $\lambda \in \sigma(A_+)\setminus
\sigma(A_-)$.
Then we have for $v_{\lambda,+}$, $v_{\lambda,-}$ as in the
 claim from above  that $v_{\lambda,+}(0)=0$ and $v_{\lambda,-}(0) \neq 0$. According to the uniqueness theorem
 $v'_{\lambda,+}(0) \neq 0$ holds. Moreover $\lambda$ is an
  isolated singularity of the function $\lambda \mapsto \frac{v'_{\lambda,+}(0)}{v_{\lambda,+}(0)}$. Recall that $v_{\lambda,+}$ depends holomorphic on $\lambda$, cf.\
\cite[Theorem 3.4.2.]{H}. But the right hand side of
\eqref{weylquotient} has no singularity at $\lambda$.
Hence there exists an open set $O$ with $O \cap \sigma_p(A)=
\emptyset$ due to the claim above. It is easy to see
that $\overline \lambda$ is an eigenvalue of $A_-$ but,
due to the fact that the opening of $Q_+$ is less than $\pi$,
cf.\ Theorem \ref{Hilfsstory}, $\overline \lambda$
is no eigenvalue of $A_+$. We obtain with the same arguments
from above $O^*:=\left\{\overline{\lambda}:\lambda \in O \right\}$ with $O^* \cap \sigma_p(A)= \emptyset$, so $(O\cup O^*)\cap \sigma_p(A)= \emptyset$.

Now assume that $\rho(A)=\emptyset$, that is, $\sigma(A)=\mathbb{C}$. If $\lambda$ is a point from the residual spectrum of $A$ (i.e., the operator $A-\lambda$ has zero kernel but
a non-dense range), then \cite[VI Theorem 6.1]{Bog}
implies  $\overline{\lambda} \in \sigma_p(A)$.
Therefore,
\begin{equation}\label{contradiction}
O \cup O^* \subset \sigma_c(A),
\end{equation}
where $ \sigma_c(A)$ denote the set of all $\lambda \in \mathbb C$ such that the operator $A-\lambda$ has zero kernel and
a dense but non-closed range.
We choose now $\lambda \in (O\cup O^*)\cap \rho(A_+)
\cap \rho(A_-)$. Then we have $\lambda \in \rho(A_-\oplus A_+)$. As $A_-\oplus A_+ \subset A_{max}$,  we see $\mathrm{ran}\,(A_{max}-\lambda)=L^2(\mathbb{R})$.
 As the minimal operator $A_0$ is the direct sum of two closed operators
 (cf.\ Theorem \ref{Hilfsstory}) it is a closed operator.
With $\rho(A_{\pm})\subset\Pi(A_{\pm})\subset \Pi(A_{0\pm})$ we get $\lambda \in \Pi(A_0)$ and from \eqref{PlusPunktPlus} we obtain
$$
 \mathrm{def}\,(A_{0}-\lambda)=2,
 $$
hence the operator $A_0-\lambda$ has a closed range. As $A_0\subset A$ and $\mathrm{dim}\,\mathrm{dom}\,A / \mathrm{dom}\,A_{0}=2$ also the range of $A-\lambda$ is closed,
a contradiction to \eqref{contradiction} and we have $\rho(A) \neq \emptyset$. Moreover we have for $\lambda \in \rho(A)\cap \rho(A_D)$
$$
\mathrm{rank}\,((A-\lambda)^{-1}-(A_D-\lambda)^{-1})\leq 2
$$
and thus the essential spectra coincide, cf.\ \cite[IX Theorem 2.4]{EE}.

According to limit-point/limit-circle classification we have
 for $\lambda \in \sigma_p(A)$
$$
\mathrm{dim}\, \mathrm{ker}\,(A-\lambda)=1.
$$
The symmetry of the spectrum follows from
Proposition \ref{symmetry}. \qed
\end{proof}

\section{Conclusion}
Summing up, our main results include
\begin{itemize}
\item[1.] A limit-point/limit-circle classification of \eqref{SLEE} and \eqref{SLEE1}, plus a mathematical meaning of Stokes wedges and Stokes lines, which is the limit-point/limit-circle classification.
\item[2.] The operator $A$,
which corresponds to the full axis problem \eqref{HAMEQ}
 with a coupling condition in zero, is self-adjoint in the inner product $[\cdot,\cdot]$ and it is $\mathcal{PT}$-symmetric.
\item[3.] The spectrum of $A$ consists at most of isolated eigenvalues with finite algebraic multiplicity, the essential spectrum is empty and $A$ has a non-empty resolvent set.
\end{itemize}

\end{document}